\documentclass[a4paper]{article}
\usepackage{amsmath}
\usepackage{amssymb}
\usepackage{amsthm}

\newtheorem{thm}{Theorem}

\newtheorem{lem}{Lemma}

\newcommand{\manifold}{\mathcal{M}}
\newcommand{\lightcone}{\mathcal{H}}
\usepackage{authblk}

\begin{document}
 \title{On the future of solutions to the massless Einstein-Vlasov system in a Bianchi I cosmology}

\author[1]{Ho Lee\footnote{holee@khu.ac.kr}}
\author[2,3]{Ernesto Nungesser\footnote{ernesto.nungesser@icmat.es}}
\author[4]{Paul Tod\footnote{tod@maths.ox.ac.uk}}
\affil[1]{Department of Mathematics and Research Institute for Basic Science, Kyung Hee University, Seoul, 02447, Republic of Korea}
\affil[2]{Instituto de Ciencias Matem\'{a}ticas (CSIC-UAM-UC3M-UCM), 28049 Madrid, Spain}
\affil[3]{Department of Mathematics, Universidad Polit\'{e}cnica de Madrid, 28040 Madrid, Spain}
\affil[4]{Mathematical Institute, University of Oxford, Oxford OX2 6GG}

\maketitle
\begin{abstract}
We show that massless solutions to the Einstein-Vlasov system in a Bianchi I space-time with small anisotropy, i.e. small shear and small trace-free part of the spatial energy momentum tensor, tend to a radiation fluid in an Einstein-de Sitter space-time with the anisotropy $\Sigma^a_b\Sigma^b_a$  and $\tilde{w}^i_j \tilde{w}^j_i$ decaying as $O(t^{-\frac12})$.
 \end{abstract}
 
 \section{Introduction}
 The massless and massive Einstein-Vlasov equations have been objects of intensive study in mathematical general relativity at least since Rendall drew attention to them in the early 1990's \cite{R0} (see the review \cite{and} for material to 2011). More recently, the stability of Minkowski space among solutions of massless \cite{tay} and massive \cite{fjs,LT} Einstein-Vlasov have been shown, and the methods of Friedrich \cite{HF} have been extended to establish existence for massless Einstein-Vlasov with data at space-like future-null-infinity $\mathcal{I}$ in the presence of positive cosmological constant ($\Lambda>0$) or data on an asymptotically hyperbolic initial surface with $\Lambda=0$, \cite{jtk}.
 
 Massless Einstein-Vlasov has better behaviour under conformal rescaling than massive. This was exploited in \cite{jtk}, and in the opposite direction in \cite{A3} and \cite{AT2} to establish local existence with data given at a \emph{conformal gauge} (or \emph{isotropic}) cosmological singularity. In \cite{HL} long-time existence and asymptotic behaviour for massive Einstein-Vlasov with $\Lambda>0$ in the spatially homogeneous setting was established, and this was extended to the massless case by \cite{T}. The late-time asymptotics for massive homogeneous solutions of the Einstein-Vlasov system are well studied close to self-similar and non self-similar solutions \cite{FH,LN2,E4}. In \cite{R1} a set of conditions was given for long-time existence of various matter models with spatial homogeneity, and massive Vlasov was identified as a case satisfying these conditions. It is possible to verify that massless Vlasov satisfies conditions (1),(2) and (7) in \cite{R1}, so that this model should also allow a proof of long-time existence even if $\Lambda=0$. That is what we consider in this article, in the first instance restricting to Bianchi type I. Note that this is not a special case of \cite{jtk} as here there is no assumption of asymptotic flatness. However this assumption means that, with the restriction to a spatially-homogeneous distribution function, the Vlasov equation is automatically solved by a distribution function $f(p_1,p_2,p_3)$ where the $p_i$ are conserved quantities for the geodesic equation obtained from the three translational Killing vectors. Assuming that the anisotropy of both metric and energy-momentum tensor is small we are able to show that the space-time tends to a radiation fluid in an Einstein-de Sitter space-time. In other words the space-time isotropises in the metric and the energy-momentum tensor. These results are made precise in Theorem 1 in Section 3. That this happens was already shown for Bianchi I LRS  \cite{RT} and Bianchi I reflection symmetric space-times \cite{R,HU1}. In \cite{BFH} decay rates for the reflection symmetric case were obtained and a formalism developed which will likely to be useful in the non-homogeneous case.
 Here we extend the former results to the case where no additional symmetries other than Bianchi I are imposed and we also obtain the optimal decay rate at which this happens.

 \section{The massless Einstein-Vlasov system}
In this section we introduce the massless Einstein-Vlasov system. Consider a four-dimensional oriented and time oriented Lorentzian manifold $(\manifold, {^4g})$ and a distribution function $f$, then the massless Einstein-Vlasov system is written as
\begin{align*}
G_{\alpha\beta}&= T_{\alpha \beta},\\
\mathcal{L} f&=0,
\end{align*}
where $\mathcal{L}$ is the Liouville operator, $G_{\alpha\beta}$ is the Einstein tensor and $T_{\alpha\beta}$ is the energy-momentum tensor defined by
\begin{align*}
T_{\alpha\beta}=  \int_{\lightcone 
\setminus  \{0\} } f p_{\alpha} p_{\beta}\omega_p.
\end{align*}
Here, the integration is over the future pointing light-cone $\lightcone$ at a given space-time point which is defined by
\[
p_{\alpha} p_{\beta}g^{\alpha\beta}=0, \ p^0>0
\]
with the apex removed, and $\omega_p$ is the Lorentz invariant measure on the light-cone of $p$. The basic equations we shall use can be found in Sections 7.3--7.4 and Chapter 25 of \cite{Ring}. We also refer to this book for an introduction to the Einstein-Vlasov system. Let $\Sigma$ be a spacelike hypersurface in $\manifold$ with $n$ its future directed unit normal. We define the second fundamental form as $k(X,Y)=\langle \nabla_X n, Y\rangle$ for vectors $X$ and $Y$ tangent to $\Sigma$, where $\nabla$ is the Levi-Civita connection of $^4g$. Here and throughout the paper we assume that Greek letters run from $0$ to $3$, while Latin letters from $1$ to $3$, and also follow the sign conventions of \cite{Ring}.

The Hamiltonian and momentum constraint equations are, as usual:
\begin{align*}
&\overline{R}-{k}_{ij} {k}^{ij}+ {k}^2=2 \rho,\\
&\overline{\nabla}^j {k}_{ji}-\overline{\nabla}_i  {k}= -{J}_i,
\end{align*}
where ${g}$ is the induced metric on $\Sigma$, $k=k_{ab}g^{ab}$ the trace of the second fundamental form, $\overline{R}$ and $\overline{\nabla}$ the scalar curvature and the Levi-Civita connection of ${g}$ respectively, and matter terms are given by $\rho= T_{\alpha\beta}n^{\alpha}n^{\beta}$ and $J_i X^i=- T_{\alpha\beta}n^{\alpha}X^{\beta}$ for $X$ tangent to $\Sigma$.

\subsection{The massless Einstein-Vlasov system with Bianchi I symmetry}

 The metric of a Bianchi spacetime in a left-invariant frame is written as
\begin{align*}
^4 g =-dt \otimes dt + g_{ij}(t)\xi^i \otimes \xi^j
\end{align*}
on $\manifold=I \times G$ with $e_0$ future oriented. We will need equations (25.17)--(25.18) of \cite{Ring} (without scalar field) with the notation $T_{ab}=S_{ab}$ and since the 3-metric is flat for Bianchi I with $R_{ab}=R=0$:
\begin{align*}
&\dot{g}_{ab}=2k_{ab}, \\
&\dot{k}_{ab}=2 k^i_a k_{bi} -k\, k_{ab}+S_{ab},
\end{align*}
where the dot means the derivative with respect to time $t$. Note that in the massless case $g^{ab} S_{ab}=\rho$.
Since $k_{ab}=k_{ab}(t)$ the constraint equations are as follows: 
\begin{align}
&2\rho=-{k}_{ij} {k}^{ij}+ {k}^2,\label{CE1} \\
&{J}_i=0.
\end{align}
The time-derivative of the mixed version of the second fundamental form is given by
\begin{align}\label{MV}
\dot{k}^a_b= -k\,k^a_b +S^a_b,
\end{align}
and by taking the trace of (\ref{MV}) we have
\begin{align}\label{im}
\dot k=-k^2+\rho.
\end{align}
It is convenient to express the second fundamental form as
\begin{align*}
k_{ab}=\sigma_{ab}+H g_{ab},
\end{align*}
where $\sigma_{ab}$ is the trace free part and $H=\frac13 k$ is the Hubble parameter. Then \eqref{im} becomes
\begin{align*}
\dot{H}=-3H^2+\frac13 \rho,
\end{align*}
and \eqref{CE1} becomes, introducing $\Omega$:
\begin{align}\label{hconstraint}
\Omega:= \frac{\rho}{3H^2}= 1-\frac16 F,
\end{align}
where $F=\frac{\sigma_{ab}\sigma^{ab}}{H^2}$. In terms of the trace free part \eqref{MV} transforms into
\begin{align*}
\dot{\sigma}^a_b= -3H \sigma^a_b  +S^a_b-(3H^2+\dot{H})\delta^a_b= -3H \sigma^a_b  +\pi^a_b,
\end{align*}
where $\pi^a_b$ is the trace free part of $S^a_b$. By \eqref{CE1} one can eliminate the energy density such that (\ref{im}) reads:
\begin{align}\label{in}
\dot k= -\frac12 k^2-\frac12 k_{ij}k^{ij}.
\end{align}
Using the trace free part of $k_{ab}$, the Hubble variable and \eqref{hconstraint} we obtain
\begin{align}\label{evH}
\frac{d}{dt}\left(\frac{1}{H}\right)=-\frac{\dot{H}}{H^2}=2 + \frac16 F=3-\Omega.
\end{align}
It is convenient to introduce a dimensionless time variable $\tau$ as follows:
\begin{align}\label{deftau}
    \frac{dt}{d\tau}=H^{-1},
\end{align}
and denote derivation with respect to that variable by a prime.  Sometimes it is also useful to use the variable $q$:
\begin{align*}
    q=-1-\frac{\dot{H}}{H^2}= 1+\frac16 F,
\end{align*}
where we have used \eqref{evH} in the last equation.
The evolution equation for $\Sigma_a^b=\frac{\sigma_a^b}{H}$ is
\begin{align*}
({\Sigma}^a_b)'=-\left(3+\frac{\dot{H}}{H^2}\right)\Sigma^a_b+\frac{S^a_b-\frac13\delta^a_b\rho}{H^2}.
\end{align*}
Using \eqref{evH} and \eqref{hconstraint} we have
\begin{align}\label{sigma+}
({\Sigma}^a_b)'=-\Omega \left(\Sigma^a_b-3w^a_b+\delta^a_b\right),
\end{align}
with $w^a_b=\frac{S^a_b}{\rho}$.
Since $\Sigma^a_b$ is trace free it is convenient to work with $\Sigma_+$ and $\Sigma_-$ defined by
\begin{align*}
\Sigma_+=\frac{1}{2H}\left (\sigma^2_2+\sigma^3_3\right), \ \Sigma_-=\frac{1}{2\sqrt{3}H}\left (\sigma^2_2-\sigma^3_3\right).
\end{align*}
The evolution equations for $\Sigma_-$ and $\Sigma_+$ can be found in \cite{LN2}. We use the constraint equation to substitute $\frac{1}{H^2}$ in $S_\pm$ and define $w_\pm$ analogously to $\Sigma_{\pm}$ by
\begin{align*}
    w_+ = \frac{\pi^2_2+\pi^3_3}{2\rho},\ w_- = \frac{\pi^2_2-\pi^3_3}{2\sqrt{3}\rho},
\end{align*}
Our equation are thus
\begin{align}
\label{plus}&{\Sigma}_+'=(q-2)\Sigma_++3w_+\Omega,\\
\label{minus}&{\Sigma}_-'=(q-2)\Sigma_-+3w_-\Omega.
\end{align}

\subsection{The Vlasov equation with Bianchi symmetry}

Since we use a left-invariant frame, and assume spatial homogeneity, $f$ will not depend on $x^a$. Moreover for Bianchi I the components $p_a$ of momenta in the invariant frame are constants of motion for the geodesic equation, so that the distribution function can be taken to be $f(p_1,p_2,p_3)$ with no explicit time dependence. We will assume that $f$ has compact support in momentum for simplicity. Since $g_{00}=g^{00}=-1$ and $g^{0a}=0$, we have $p^0=-p_0=\sqrt{p_ap_bg^{ab}}$, $\rho=T_{00}$, and $J_{a}=-T_{0a}=0$. The frame components of the energy-momentum tensor are thus
\begin{align*}
\label{rho}&\rho=(\det g)^{-\frac12}  \int f p^0 dp,\ 
&S_{ij}=(\det g)^{-\frac12} \int f \frac{p_i p_j} {p^0}dp,
\end{align*}
where the distribution function is understood as $f=f(p)$ with $p=(p_1,p_2,p_3)$. The Vlasov equation is simply
\begin{align}
  \frac{\partial f}{\partial t} = 0.
\end{align}

Consider the derivative of the following quantity:
\begin{align}\label{w}
w^j_i = \frac{S^j_i}{\rho}=\frac{\int f p_i p_a g^{aj} (p^0)^{-1}dp}{\int f p^0 dp}.
\end{align}
The derivative of $w_i^j$ with respect to $\tau$ is
\begin{align}\label{evw0}
(w_i^j)'=-2w_i^a \Sigma^j_a+w^b_a \Sigma^a_b w_i^j+ \Sigma^{c}_d\xi_{ic}^{jd},
\end{align}
where
\begin{align}\label{xi}
\xi_{ic}^{jd}=\frac{\int f p_ip^jp_cp^d (p^0)^{-3}dp}{\int fp^0 dp}.
\end{align}
 The derivative of ${\xi}_{ic}^{jd}$ is
\begin{align}
\left({\xi}_{ic}^{jd}\right)'=-2\Sigma^d_f {\xi}_{ic}^{jf}-2\Sigma^j_f{\xi}_{ic}^{fd}+3\Sigma^a_b \frac{\int {f} p_ip^jp_cp^dp_ap^b (p^0)^{-5} dp}{\int {f} p^0 dp}+ {\xi}_{ic}^{jd} w_a^b \Sigma^a_b.
\end{align}

As one can see from the last equation, our system is not closed. The derivative of ${\xi}_{ic}^{jd}$ involves higher momentum terms.

\section{Main result}
We will consider that we are close to the case that $g_{ab}$ and $f$ are isotropic. In that case $\xi_{ic}^{jd}$ takes some specific values which we will denote by $\hat{\xi}_{ic}^{jd}$
Let us consider the trace free part $\tilde{w}^i_j=w^i_j-\frac13\delta_j^i$ and define $\tilde{\xi}_{ic}^{jd}=\xi_{ic}^{jd}-\hat{\xi}_{ic}^{jd}$. Then  \eqref{evw0}  turns into
\begin{align}
(\tilde{w}_i^j)'=-\frac23\Sigma^j_i-2\tilde{w}_i^a\Sigma^j_a+\tilde{w}^b_a \Sigma^a_b (\tilde{w}_i^j+\frac13\delta^j_i)+ \Sigma^{c}_d \left(\tilde{\xi}_{ic}^{jd}+\hat{\xi}_{ic}^{jd}\right).
\end{align}
For the diagonal terms we have
\begin{align*}
&w_+' =-\frac23 \Sigma_+- \tilde{w}_2^a\Sigma^2_a - \tilde{w}_3^a\Sigma^3_a +\tilde{w}^b_a \Sigma^a_b \left(w_++\frac13\right)+\frac12 \Sigma^{c}_d \left(\tilde{\xi}_{2c}^{2d}+\hat{\xi}_{2c}^{2d}+\tilde{\xi}_{3c}^{3d}+\hat{\xi}_{3c}^{3d}\right),\\
&w_-' =-\frac23 \Sigma_-- \frac{1}{\sqrt{3}}(\tilde{w}_2^a\Sigma^2_a - \tilde{w}_3^a\Sigma^3_a) +\tilde{w}^b_a \Sigma^a_b w_- +\frac{1}{2\sqrt{3}} \Sigma^{c}_d \left(\tilde{\xi}_{2c}^{2d}+\hat{\xi}_{2c}^{2d}-\tilde{\xi}_{3c}^{3d}-\hat{\xi}_{3c}^{3d}\right).
\end{align*}

In the isotropic case $g_{ij}=R^{2}\delta_{ij}$. The factor $R^2$ will disappear in that case in $\xi_{ic}^{jd}$  due to cancellation, which means the expression is basically an Euclidean one. Denoting by a hat the isotropic case we have 
\begin{align*}
\hat{\xi}_{ic}^{jd}= \delta^{ja}\delta^{db} \frac{\int \hat{f} p_ip_ap_cp_b \vert p \vert^{-3} dp}{\int \hat{f} \vert p \vert dp},
\end{align*}
with $\vert p \vert = \sqrt{p_a p_b \delta^{ab}}$.
For that case the different integrals which can appear are 
\begin{align*}
  &I_0=  \int \hat{f} \vert p\vert dp,\ I_1=  \int \hat{f} p_1^4 \vert p\vert ^{-3} dp=\frac15I_0, I_3= \int \hat{f} p_1^2p_2^2 \vert p\vert ^{-3} dp=\frac{I_0}{15}\\
  &I_2=  \int \hat{f} p_1^3p_2 \vert p\vert ^{-3} dp=0,\ I_4= \int \hat{f} p_1^2p_2p_3 \vert p\vert ^{-3} dp= 0.
\end{align*}
Suspending the Einstein summation convention for the next two equations, the only non-vanishing expressions are 
\begin{align*}
&\hat{\xi}_{aa}^{aa}=\frac{1}{5}, \ a=1,2,3\\ &\hat{\xi}_{ab}^{ab}=\hat{\xi}_{ab}^{ba}=\hat{\xi}_{aa}^{bb}=\frac{1}{15}, a\neq b, a,b=1,2,3.
\end{align*}

To summarize, our system is
\begin{align}
\label{evsigma}&({\Sigma}^a_b)'=-\Omega \left(\Sigma^a_b-3\tilde{w}^a_b\right),\\
\label{evw}&(\tilde{w}_i^j)'=-\frac23\Sigma^j_i-2\tilde{w}_i^a\Sigma^j_a+\tilde{w}^b_a \Sigma^a_b (\tilde{w}_i^j+\frac13\delta^j_i)+ \Sigma^{c}_d \left(\tilde{\xi}_{ic}^{jd}+\hat{\xi}_{ic}^{jd}\right).
\end{align}
For the diagonal terms we have
\begin{align}
    &{\Sigma}_+'=\Omega \left(3w_+-\Sigma_+\right),\\
    &{\Sigma}_-'=\Omega \left(3w_--\Sigma_-\right).
\end{align}

\begin{lem}\label{basic}
Consider the massless Einstein-Vlasov system with Bianchi I symmetry. Suppose $\Sigma^a_b(\tau_0)=\epsilon_\Sigma$, $\tilde{w}_i^j(\tau_0)=\epsilon_w$, and $\tilde{\xi}_{ic}^{jd}(\tau_0)=\epsilon_\xi$  are small initially. Then the following estimates hold:
\begin{align}
\Sigma^a_b=O(\epsilon e^{-\frac12\tau}),\\
\tilde{w}_i^j=O(\epsilon e^{-\frac12\tau}),
\end{align}
where $\epsilon$ only depends on $\epsilon_\Sigma$ and $\epsilon_w$.
\end{lem}
\begin{proof}
Let us assume for an interval $[\tau_0,\tau_1)$ that
\begin{align}
    &\label{boot1}\Sigma^a_b=\epsilon e^{-\frac14\tau},\\
&\label{boot2}\tilde{w}_i^j=\epsilon e^{-\frac14\tau}.
\end{align}

From \eqref{xi} we have the following bound using the fact the higher order momenta terms are bounded by $\rho$ since $p_ap^b\leq C(p^0)^2$:
\begin{align}\label{bound1}
\frac{\int {f} p_ip^jp_cp^dp_ap^b (p^0)^{-5} dp}{\int {f} p^0 dp} \leq C \frac{\int {f} p_ip^jp_cp^d (p^0)^{-3} dp}{\int {f} p^0 dp} = C {\xi}_{ic}^{jd}.
\end{align}
Similarly ${\xi}_{ic}^{jd}$ is also bounded since
\begin{align}\label{bound2}
    {\xi}_{ic}^{jd} \leq C \frac{\int {f} p_ip^j(p^0)^{-1} dp}{\int {f} p^0 dp}\leq C.
\end{align}

Using assumption \eqref{boot1} for $\Sigma^a_b$ and \eqref{bound1}-\eqref{bound2} that
\begin{align*}
\left({\xi}_{ic}^{jd}\right)'\leq C \vert \Sigma^a_b \vert \leq \epsilon e^{-\frac14\tau},
\end{align*}
for all indices. The same bound holds for $\left({\tilde{\xi}}_{ic}^{jd}\right)'$. Integrating we obtain that
\begin{align} \label{bound3}
{\tilde{\xi}}_{ic}^{jd}(\tau) \leq {\tilde{\xi}}_{ic}^{jd}(\tau_0)+C \epsilon.
\end{align}
This means that this quantity does not necessarily get smaller, but we can bound it by a small quantity.

If we focus on the first order terms in \eqref{evsigma} and \eqref{evw} we have that
\begin{align}
&({\Sigma}^a_b)'=-\Sigma^a_b+3\tilde{w}^a_b+O_{\Sigma},\\
&(\tilde{w}_i^j)'=-\frac23\Sigma^j_i+ \Sigma^{c}_d \hat{\xi}_{ic}^{jd}+O_{w},
\end{align}
with 
\begin{align*}
&O_{\Sigma}=\frac16\Sigma^a_b\Sigma^b_a\left(\Sigma^a_b-3\tilde{w}^a_b\right),\\
&O_{w}=-2\tilde{w}_i^a\Sigma^j_a+\tilde{w}^b_a \Sigma^a_b (\tilde{w}_i^j+\frac13\delta^j_i)+\Sigma^{c}_d \tilde{\xi}_{ic}^{jd}.
\end{align*}
The diagonal terms for $\Sigma^a_b$ can be expressed by $\Sigma_{\pm}$, so that
\begin{align}
    &{\Sigma}_+'=-\Sigma_+ + 3w_+ + O_{\Sigma},\\
    &{\Sigma}_-'=-\Sigma_- + 3w_- + O_{\Sigma}.
\end{align}

Since most of the quantities $\hat{\xi}_{ic}^{jd}$ vanish we have simple expressions. The expressions can be separated into diagonal and non diagonal terms.
For $w_{\pm}$ we obtain
\begin{align}
    &w_+'=-\frac{8}{15}\Sigma_+ + O_{w},\\
    &w_-'=-\frac{8}{15}\Sigma_- + O_{w}.
\end{align}
This means we have the same system for $\Sigma_{\pm}w_{\pm}$, namely

\begin{align}
\left(
\begin{matrix}
\Sigma_{\pm}\\
w_{\pm}
\end{matrix}
  \right)'
  =
  \left(
\begin{matrix}
-1 & 3 \\
-\frac{8}{15} & 0 
\end{matrix}
  \right)
  \left(
\begin{matrix}
\Sigma_{\pm}\\
w_{\pm}
\end{matrix}
  \right)
  +   \left(
\begin{matrix}
O_{\Sigma}\\
O_w
\end{matrix}
  \right),
\end{align}
The non-diagonal terms satisfy for any $i\neq j$ 
\begin{align}
\left(
\begin{matrix}
\Sigma^j_i\\
\Sigma^i_j\\
\tilde{w}^j_i\\
\tilde{w}^i_j
\end{matrix}
  \right)'
  = 
  \left(
\begin{matrix}
-1 & 0 & 3 & 0\\
0 & -1 & 0 &3 \\
-\frac{9}{15} & \frac{1}{15} & 0 & 0 \\
\frac{1}{15} & -\frac{9}{15} & 0 & 0 
\end{matrix}
\right)
\left(
\begin{matrix}
\Sigma^j_i \\
\Sigma^i_j \\
\tilde{w}^j_i\\
\tilde{w}^i_j\\
\end{matrix}
  \right)
  +
\left(
\begin{matrix}
O_{\Sigma}\\
O_{\Sigma}\\
O_w\\
O_w
\end{matrix}
  \right).
\end{align}

 As one can see from this, except for the error term, the diagonal and the non-diagonal part are independent.
 
 Consider the diagonal case first. The eigenvalues of the linearised system are:
 \begin{align*}
   &  \lambda = \frac{1}{10}(-5 \pm i 3\sqrt{15}).
 \end{align*}

Making the canonical transformation to
\begin{align}
\mathbf{Y}=
  \left( 
  \begin{matrix}
  \bar{\Sigma}_{\pm}\\
  \bar{w}_{\pm}
  \end{matrix}
  \right) = P^{-1} \left( 
  \begin{matrix}
  {\Sigma}_{\pm}\\
  {w}_{\pm}
  \end{matrix}
  \right),
\end{align}
with
\begin{align}
P^{-1}= \left(
\begin{matrix}
-\frac{16}{135}\sqrt{15} & \frac19 \sqrt{15}\\
0 &1
\end{matrix}
\right)
\end{align}
we obtain
\begin{align}
    \mathbf{Y}'= \frac{1}{10} \left( \begin{matrix} -5 & -3\sqrt{15}\\
    3\sqrt{15} & -5\end{matrix}\right) \mathbf{Y} +\left(
\begin{matrix}
O_{\Sigma}+ O_w\\
O_w
\end{matrix}
  \right).
\end{align}
Defining
\begin{align*}
    R= \Vert {\bf Y} \Vert^2,
\end{align*}
we obtain
\begin{align*}
R' = -R + \|\mathbf{Y} \|(O_w + O_{\Sigma}).
\end{align*}
Introducing spherical coordinates for $\bar{\Sigma}_{\pm}$ and $\bar{w}_{\pm}$ we obtain the desired estimate for the diagonal components up to an epsilon, since the error term $\|\mathbf{Y} \|(O_w + O_{\Sigma})$ is of third order and we can choose the initial value for ${\xi}_{ic}^{jd}(\tau_0)$ smaller than $\epsilon_\Sigma$. Choosing $\epsilon_\Sigma$ and $\epsilon_w$ small enough we have improved the bootstrap assumptions \eqref{boot1}-\eqref{boot2} for the diagonal components. For the non-diagonal components we obtain the eigenvalues
\begin{align*}
   & \lambda= \frac12(-1 \pm i \sqrt{7}),  \frac{1}{10}(-5 \pm 3 i \sqrt{15}).
\end{align*}
 Proceeding as for the diagonal components we obtain the desired estimates for the 
 non-diagonal ones up to an $\epsilon$ as well. Now we can do another loop as for instance in Section 5 of \cite{EN} to get rid of the epsilon and thus to obtain the desired estimates.
\end{proof}

Based on the previous lemma we can obtain estimates of other variables.
\begin{lem}\label{basic2}
Consider initial data which correspond to a massless solution of the Einstein-Vlasov system with Bianchi I symmetry which expands initially, i.e. $H(t_0)>0$ and satisfies the conditions of Lemma \ref{basic}, then
\begin{align}
   & H=\frac12t^{-1}(1+O(\epsilon^2 t^{-\frac12})),\\
    &g_{ab}=t (G_{ab}+O(\epsilon t^{-\frac14})),\\
    &g^{ab}=t^{-1}(G^{ab}+O(\epsilon t^{-\frac14})),
\end{align}
where $G_{ab}$ and $G^{ab}$ are constant matrices and $\epsilon$ a small quantity.
\end{lem}
\begin{proof}
From \eqref{evH} we obtain
\begin{align}
    \frac{1}{H}\leq 3(t-t_0)+\frac{1}{H(t_0)}.
\end{align}
Using \eqref{deftau} and setting $t_0=1/(2H(t_0))$ we obtain
\begin{align}
\frac{dt}{d\tau} \leq 3t
\end{align}
Integrating and doing some computations we obtain
\begin{align}\label{taut}
e^{-\tau}\leq C t^{-\frac13}.
\end{align}
From \eqref{evH} we also obtain
\begin{align}
H=\frac12t^{-1}\frac{1}{1+\frac12 t^{-1}I},
\end{align}
with
\begin{align}
I= \frac16 \int_{t_0}^t \Sigma_a^b \Sigma_b^a(s) ds.
\end{align}
Using now \eqref{taut} we obtain
\begin{align}
I \leq C\epsilon^2 t^{\frac23}.
\end{align}
As a result we obtain the estimate for $H$
\begin{align}
H=\frac12 t^{-1}\left(1+O\left(\epsilon^2 t^{-\frac13}\right)\right).
\end{align}
Using this in \eqref{deftau} again we obtain
\begin{align*}
e^{-\tau} \leq Ct^{-\frac12}.
\end{align*}
Doing another loop of the computations we obtain the desired estimate of $H$.
To obtain the estimate of $g_{ab}$ we do the same computations as in \cite{LN} concerning the estimate of the metric.
Note that
\begin{align}
    &\vert H-\frac12 t^{-1} \vert \leq C \epsilon^2 t^{-\frac32}\\
    &(\sigma_{ab}\sigma^{ab})^{\frac12}\leq C\epsilon t^{-\frac54},
\end{align}
and the estimate for $g_{ab}$ and $g^{ab}$ follow by the same computations as in \cite{LN}.
\end{proof}
Let us summarize our results in the following theorem:
\begin{thm}
Consider initial data which correspond to a massless solution of the Einstein-Vlasov system with Bianchi I symmetry which expands initially, i.e. $H(t_0)>0$. Suppose $\Sigma^a_b(\tau_0)=\epsilon_\Sigma$, $\tilde{w}_i^j(\tau_0)=\epsilon_w$, and $\tilde{\xi}_{ic}^{jd}(\tau_0)=\epsilon_\xi$  are small initially. Then the following estimates hold:
\begin{align}
&\Sigma^j_i=O(\epsilon t^{-\frac14}),\\
&\tilde{w}^i_j=O(\epsilon t^{-\frac14}),\\
   & H=\frac12t^{-1}(1+O(\epsilon^2 t^{-\frac12})),\\
    &g_{ab}=t (G_{ab}+O(\epsilon t^{-\frac14})),\\
    &g^{ab}=t^{-1}(G^{ab}+O(\epsilon t^{-\frac14})),
\end{align}
where $\epsilon$ only depends on the initial data.
\end{thm}

\section{Acknowledgements}
The authors thank Prof John Stalker of Trinity College Dublin for helpful discussions. H.L. and P.T. acknowledge financial support from ICMAT. H.L. was supported by the Basic Science Research Program through the National Research Foundation of Korea (NRF) funded by the Ministry of Science, ICT \& Future Planning (NRF-2018R1A1A1A05078275). EN acknowledges support from grants MTM2017-86875-C3-1-R AEI/ FEDER, UE and RTC-2017-6593-7 AEI/FEDER, UE.


\begin{thebibliography}{10}
\bibitem{and}
H. Andr\'easson.
\newblock{The Einstein-Vlasov System/Kinetic Theory.}
\newblock{\em Living Reviews in Relativity}, https://doi.org/10.12942/lrr-2011-4. 

\bibitem{A3}K. Anguige.
\newblock{Isotropic cosmological singularities. III. The Cauchy problem for the inhomogeneous conformal Einstein-Vlasov equations.}
\newblock{\em Ann. Physics} 282 395--419, 2000.

\bibitem{AT2} K. Anguige and K. P. Tod.
\newblock{Isotropic cosmological singularities. II. The Einstein-Vlasov system.}
\newblock{\em Ann. Physics} 276 294--320, 1999.


\bibitem{BFH}
H.Barzegar, D. Fajman, G. Hei{\ss}el.
\newblock{On slowly expanding spacetimes}
\newblock{arXiv:1904.13290}
 
\bibitem{fjs} 
D. Fajman, J. Joudioux and J. Smulevici.
\newblock{The Stability of the Minkowski space for the Einstein-Vlasov system.}
\newblock{arXiv:1707.06141}


\bibitem{FH}
D.~Fajman and G.~Hei{\ss}el.
\newblock{Kantowski-Sachs cosmology with Vlasov matter}
\newblock {\em Class. Quant. Grav.} 36: 135002, 2019.

\bibitem{HF}
H. Friedrich.
\newblock{On the existence of n-geodesically complete or future complete solutions of Einstein’s field equations with smooth asymptotic structure}
\newblock{\em Comm.Math. Phys.} 107 587--609, 1986.

\newblock{and}

\newblock{On the global existence and the asymptotic behavior of solutions to
the Einstein-Maxwell-Yang-Mills equations,}
\newblock{\em J. Diff. Geom.} 34, 275--345, 1991.

\bibitem{HU1}
J.M.~Heinzle and C.~Uggla.
\newblock {Dynamics of the spatially homogeneous Bianchi type I Einstein-Vlasov equations}.
\newblock {\em Class. Quant. Grav.} 23, 3463--3490, 2006.



\bibitem{jtk}
J. Joudioux, M. Thaller and J.A. Valiente Kroon.
\newblock{The Conformal Einstein Field Equations with Massless Vlasov Matter.}
\newblock{arXiv:1903.12251.}

\bibitem{HL} H.~Lee.
\newblock{Asymptotic behaviour of the Einstein-Vlasov system with a positive cosmological constant.}
\newblock{\em Math. Proc. Cambridge Philos. Soc.} 137, 495--509, 2004. 

\bibitem{LN}
H.~Lee and E.~Nungesser.
\newblock {Future global existence and asymptotic behaviour of solutions to the
  Einstein-Boltzmann system with Bianchi I symmetry}.
\newblock {\em J. Differ. Equations} 262, 11:5425-5467,2017.

\bibitem{LN2}
H.~Lee and E.~Nungesser.
\newblock {Self-similarity breaking of cosmological solutions with collisionless matter}.
\newblock {\em Ann. Henri Poincare} 19, 7:2137–-2155, 2018.

\bibitem{LT}
H. Lindblad and M.Taylor.
\newblock{Global stability of Minkowski space for the Einstein--Vlasov system in the harmonic gauge.}
\newblock{arXiv:1707.06079} 



\bibitem{EN}
E.~Nungesser.
\newblock {Isotropization of non-diagonal Bianchi I spacetimes with
  collisionless matter at late times assuming small data}.
\newblock {\em Class. Quant. Grav.} 27:235025, 2010.


\bibitem{E4}
E.~Nungesser.
\newblock {Future non-linear stability for solutions of the Einstein-Vlasov system of Bianchi types II and VI$_0$}.
\newblock {\em J. Math. Phys.} 53, 102503, 2012.



\bibitem{R0}
A.D. Rendall.
 \newblock{On the choice of matter model in general relativity.}
 \newblock{pp 94--102 in Approaches to numerical relativity (Southampton, 1991)}
 \newblock{Cambridge Univ. Press, Cambridge, 1992.} 

\bibitem{R1}
A. D. Rendall.
\newblock{Global properties of locally spatially homogeneous cosmological models with matter}
\emph{Math. Proc. Camb. Phil. Soc.} 118: 511--526, 1995.

\bibitem{R} 
A.D.Rendall.
\newblock{The Initial singularity in solutions of the Einstein-Vlasov system of Bianchi type I}
\newblock{\em J. Math. Phys.} 37:438--451, 1996.


\bibitem{RT}
A.D. Rendall and K.P. Tod.
\newblock {Dynamics of spatially homogeneous solutions of the Einstein-Vlasov equations which are locally rotationally symmetric}
\emph{\em Class. Quant. Grav.} 16: 1705--1726, 1998.

     
\bibitem{Ring}
H.~Ringstr{\"{o}}m.
\newblock {\em {On the Topology and Future Stability of the Universe}}.
\newblock Oxford University Press, Oxford, 2013.

\bibitem{tay}
M. Taylor .
\newblock{The global nonlinear stability of Minkowski space for the massless Einstein--Vlasov system.}
\newblock{arXiv:1602.02611} 

\bibitem{T}
 P. Tod.
 \newblock{Isotropic cosmological singularities in spatially homogeneous models with a cosmological constant.}
 \newblock{\em Class. Quant. Grav.} 24, 2415--2432, (2007). 

\end{thebibliography}
\end{document}